\documentclass[11pt]{article}
\usepackage{amsmath,amssymb,amsthm,mathtools}
\usepackage{geometry}
\usepackage{hyperref}
\usepackage{microtype}
\usepackage{enumitem}
\title{One-Shot and Concurrent Hitting Times for Grover-Coined Quantum Walks on  Cubelike Graphs}
\author{Jaideep Mulherkar}
\date{\today}

\newtheorem{theorem}{Theorem}
\newtheorem{lemma}[theorem]{Lemma}
\newtheorem{proposition}[theorem]{Proposition}

\theoremstyle{definition}

\DeclareMathOperator{\Cay}{Cay}
\newcommand{\Z}{\mathbb{Z}}
\newcommand{\C}{\mathbb{C}}
\newcommand{\Id}{\mathrm{I}}
\newcommand{\ket}[1]{\lvert #1\rangle}
\newcommand{\bra}[1]{\langle #1\rvert}
\newcommand{\braket}[2]{\langle #1\mid #2\rangle}

\begin{document}
\maketitle
\begin{abstract}
We study the one-shot and concurrent hitting for the discrete-time
Grover-coined quantum walk on cubelike graphs
$G=\operatorname{Cay}(\mathbb Z_2^d,\Omega)$ of degree
$\Delta=|\Omega|$. Starting from the vertex labeled $0$, we identify
$\sigma=\bigoplus_{\omega\in\Omega}\omega$ as a natural target vertex;
for the hypercube, $\sigma$ is precisely the antipodal vertex.

For families with $\Delta\to\infty$, let $T$ be an integer having the
same parity as $\Delta$ and satisfying
\[
\left|T-\frac{\pi\Delta}{2}\right|\leq 1.
\]
We show that the probability $p_T(\sigma)$ of finding the walker at
$\sigma$ when it is measured at time $T$ satisfies
\[
p_T(\sigma)=1-O(\Delta^{-1/5}).
\]
Thus the target is found with probability tending to one after
$\Theta(\Delta)$ steps.

For the concurrently measured walk, let
$H_T^{\mathrm{Conc}}(\sigma)$ denote the probability that the target is
detected at or before time $T$ when it is tested after every step. We
prove
\[
p_T(\sigma)\leq T H_T^{\mathrm{Conc}}(\sigma),
\]
which implies
$H_T^{\mathrm{Conc}}(\sigma)=\Omega(\Delta^{-1})$
over the same time scale.

The proof uses the Walsh-Fourier decomposition, an exact
two-dimensional reduction of each Fourier mode, and a universal
second-moment identity for the associated character sums. Our results
extend Kempe's hypercube hitting phenomenon \cite{Kempe2005} to
arbitrary cubelike generating sets and establish the conjectured
asymptotic hitting behavior for cubelike and augmented cubes in
\cite{MulherkarRajdeepakSunitha2022}.
\end{abstract}
\section{Introduction}

Quantum walks are quantum analogues of classical random walks and have
become an important framework in quantum information, quantum algorithms,
and the study of transport on graphs. They occur in continuous-time and
discrete-time forms. Continuous-time quantum walks were introduced in the
computational setting by Farhi and Gutmann \cite{FarhiGutmann1998}, while
discrete-time quantum walks on graphs were developed systematically by
Aharonov, Ambainis, Kempe, and Vazirani \cite{Aharonov2001}. For general
introductions and reviews, see
\cite{Kempe2003,VenegasAndraca2012,Portugal2013}. Quantum walks have led
to algorithms for search \cite{Shenvi2003}, element distinctness
\cite{Ambainis2007}, and have also been shown to provide a universal
model for quantum computation \cite{Childs2009}.

A basic question is the \emph{hitting problem}: starting from one vertex,
how rapidly can a quantum walk reach a specified target? The hypercube
provides a particularly striking example. Moore and Russell analyzed the
spectral structure and mixing behavior of quantum walks on the hypercube
\cite{MooreRussell2002}. Kempe subsequently showed that a discrete-time
Grover-coined quantum walk can reach the antipodal vertex on a time scale
linear in the dimension, in contrast with the exponentially large
classical hitting time \cite{Kempe2005}. She introduced two notions that
are relevant here:  one-shot hitting, in which the walk is measured only
at a prescribed time, and concurrent hitting, in which the target is
tested after every step.

The effect of repeated measurement on quantum hitting has subsequently
been studied from several perspectives. Krovi and Brun formulated a
measured hitting time using superoperators and analyzed the hypercube
\cite{KroviBrun2006}. More generally, repeated projective
measurements lead naturally to quantum first-detection problems
\cite{Friedman2017}. These results highlight an important distinction:
the probability of occupying a target at a prescribed time and the
probability of detecting it under repeated measurements are different
quantities and need not have the same behavior.

In this paper we extend the hypercube hitting framework to cubelike
graphs
\[
G=\operatorname{Cay}(\mathbb Z_2^d,\Omega).
\]
The hypercube corresponds to
$\Omega=\{e_1,\ldots,e_d\}$, whereas a general cubelike graph may contain
generators of different Hamming weights and need not possess the full
permutation symmetry of the hypercube. Writing
$\Delta=|\Omega|$, we identify
\[
\sigma=\bigoplus_{\omega\in\Omega}\omega
\]
as a canonical target vertex. Our main result shows that, as the degree
$\Delta$ tends to infinity, the Grover-coined moving-shift walk reaches
this vertex with probability tending to one after a number of steps
linear in $\Delta$. Thus the asymptotic hitting phenomenon known for the
hypercube extends to arbitrary cubelike generating sets under the walk
model considered here.

We also establish a direct connection between one-shot  and concurrent
hitting. If $p_T(\sigma)$ is the probability of finding the walker at
$\sigma$ when it is measured only at time $T$, and
$H_T^{\mathrm{Conc}}(\sigma)$ is the probability of detecting $\sigma$
by time $T$ when it is tested after every step, we show that the former
controls the latter through a general inequality. The argument expresses
the unmeasured amplitude at time $T$ in terms of propagated
first-detection amplitudes. Importantly, this relation does not depend on
the cubelike structure and therefore applies more generally to
discrete-time quantum walks. Combined with the  theorem, it
yields a concurrent hitting probability of order at least
$\Delta^{-1}$ over a linear time scale.

The proof of the  result isolates a general mechanism behind the
hypercube hitting phenomenon. Walsh characters diagonalize the
translations of $\mathbb Z_2^d$, and each Fourier mode of the coined walk
reduces exactly to a two-dimensional rotation. The corresponding
character sum determines the rotation angle. A universal second-moment
identity shows that all but a vanishing fraction of these character sums
are sufficiently small, while the character of the XOR target supplies
the sign needed for constructive interference at the relevant time.
This produces asymptotic phase alignment without requiring explicit
knowledge of the multiplicities of the individual Fourier modes. In
particular, the argument does not require bounded Hamming weight of the
generators, distance regularity, or permutation symmetry.

There is also a connection with a different quantum transport phenomenon
on cubelike graphs. In the continuous-time setting, Bernasconi, Godsil,
and Severini showed that the XOR of the generators plays a distinguished
role in perfect state transfer \cite{BernasconiGodsilSeverini2008}; see
also \cite{Godsil2012} for a broader treatment of state transfer on
graphs. The continuous-time and coined discrete-time models are
different, and the results should not be identified with one another.
Nevertheless, the appearance of the same algebraically distinguished
displacement is noteworthy. In the present discrete-time setting it
arises from the phase alignment of the Walsh--Fourier modes.

Our results apply in particular to augmented cubes
\cite{ChoudumSunitha2002}. In earlier numerical work, asymptotically
linear hitting behavior in the degree was conjectured for cubelike and
augmented-cube families \cite{MulherkarRajdeepakSunitha2022}. The present
work provides an analytic proof of this behavior for the Grover-coined
moving-shift walk and identifies the Fourier mechanism responsible for
it.

The paper is organized as follows. Section~2 defines the cubelike walk
and the one-shot and concurrent hitting notions. Section~3 develops the
Walsh--Fourier decomposition and the two-dimensional reduction of each
Fourier mode. Section~4 derives the target-amplitude formula and the
character-sum identities used in the analysis. Section~5 proves the
 hitting theorem. Section~6 establishes the relation between
 one shot and concurrent hitting and derives the concurrent bound. The
remaining sections discuss examples and consequences, including
hypercubes, augmented cubes, and complete cubelike graphs.

\section{Setup}
\subsection{Cubelike graph and coined walk}
Let $d\ge 1$. We write
\[
 \Z_2^d=(\Z/2\Z)^d
\]
for the elementary abelian $2$-group, with componentwise
addition modulo $2$. Let
\[
 \Omega\subseteq\Z_2^d\setminus\{0\},\qquad \Delta:=|\Omega|,
\]
and assume the elements of $\Omega$ are distinct.  The cubelike graph is
\[
 G=\Cay(\Z_2^d,\Omega),
\]
with an edge from $x$ to $x\oplus\omega$ for each $x\in\Z_2^d$ and
$\omega\in\Omega$.

The coin and position spaces are
\[
 \mathcal H_C=\operatorname{span}\{\ket{\omega}:\omega\in\Omega\}
 \cong\C^\Delta,
 \qquad
 \mathcal H_P=\operatorname{span}\{\ket{x}:x\in\Z_2^d\}.
\]
Define the uniform coin state and Grover coin by
\begin{equation}\label{eq:Dcoin}
 \ket D=\frac1{\sqrt\Delta}\sum_{\omega\in\Omega}\ket\omega,
 \qquad
 C=2\ket D\!\bra D-\Id_{\mathcal H_C}.
\end{equation}
The shift operator is defined as
\begin{equation}\label{eq:shift}
 S\ket{\omega,x}=\ket{\omega,x\oplus\omega},
\end{equation}
and one step of the walk is
\begin{equation}\label{eq:walk}
 U=S(C\otimes\Id_{\mathcal H_P}).
\end{equation}
We use the initial state
\begin{equation}\label{eq:initial}
 \ket{\psi_0}=\ket D\otimes\ket 0.
\end{equation}

\subsection{Target, one-shot and concurrent hitting times}
Define the target vertex $\sigma$ as the algebraic XOR of the elements of $\Omega$. In the hypercube this is the diametrically opposite vertex
\begin{equation}\label{eq:sigma}
 \sigma:=\bigoplus_{\omega\in\Omega}\omega.
\end{equation}
The position probability at $v\in\Z_2^d$ after $T$ steps is
\begin{equation}\label{eq:positionprob}
 p_T(v):=\sum_{\omega\in\Omega}
 \left|\bra{\omega,v}U^T\ket{\psi_0}\right|^2.
\end{equation}
Following Kempe  \cite{Kempe2005}, we define two types of hitting times
\begin{enumerate}
    \item \underline{One-shot hitting time:}
a pair $(T,p)$ is a one-shot
 hitting time from $0$ to $v$ when $p_T(v)\ge p$.\\

\item \underline{Concurrent hitting time:}
To define the concurrent hitting time we first define the measured walk. At each step of the quantum walk $U$ we measure to see if the walk is at a fixed position $\ket{v}$ or not. That is, we do a projective measurement with the projectors $P_v=\Id\otimes \ket{v}\bra{v}$ and $Q_v=\Id -P_v$. If the walker was found at $\ket{v}$ we stop otherwise we continue the walk. This is called the measured walk.

A quantum walk $U$ has a $(T,p)$ concurrent hitting time if the $\ket{v}$ measured walk has probability $\geq p$ of stopping at time $t\leq T$
\end{enumerate}

\section{Fourier decomposition}
For $a,x\in\Z_2^d$, write $a\cdot x$ for the binary inner product modulo
$2$, and define the character
\[
 \chi_a(x):=(-1)^{a\cdot x}.
\]
We have the following well-known properties
\begin{proposition}[Walsh characters and Fourier basis]
\label{prop:walsh-characters}
For $a,x\in\mathbb Z_2^d$, define the Walsh character
\[
\chi_a(x)=(-1)^{a\cdot x},
\]
where
\[
a\cdot x=\sum_{j=1}^d a_jx_j \pmod 2.
\]
The Walsh characters satisfy the following properties.

\begin{enumerate}[label=(\roman*)]

\item \textbf{Multiplicativity.}
For all $a,x,y\in\mathbb Z_2^d$,
\[
\chi_a(x\oplus y)
=
\chi_a(x)\chi_a(y).
\]

\item \textbf{Product of characters.}
For all $a,b,x\in\mathbb Z_2^d$,
\[
\chi_a(x)\chi_b(x)
=
\chi_{a\oplus b}(x).
\]

\item \textbf{Orthogonality in the group variable.}
For all $a,b\in\mathbb Z_2^d$,
\[
\frac{1}{2^d}
\sum_{x\in\mathbb Z_2^d}
\chi_a(x)\chi_b(x)
=
\delta_{a,b}.
\]

\item \textbf{Orthogonality in the character variable.}
For all $x,y\in\mathbb Z_2^d$,
\[
\frac{1}{2^d}
\sum_{a\in\mathbb Z_2^d}
\chi_a(x)\chi_a(y)
=
\delta_{x,y}.
\]

\item \textbf{Character sum.}
For every $a\in\mathbb Z_2^d$,
\[
\sum_{x\in\mathbb Z_2^d}\chi_a(x)
=
\begin{cases}
2^d, & a=0,\\
0,   & a\neq0.
\end{cases}
\]

\item \textbf{Fourier basis and inversion.}
Define the normalized Fourier basis states by
\[
|\widehat a\rangle
=
\frac{1}{2^{d/2}}
\sum_{x\in\mathbb Z_2^d}
\chi_a(x)|x\rangle,
\qquad a\in\mathbb Z_2^d.
\]
These states form an orthonormal basis,
\[
\langle\widehat a|\widehat b\rangle
=
\delta_{a,b},
\]
and the inverse Fourier transform is
\[
|x\rangle
=
\frac{1}{2^{d/2}}
\sum_{a\in\mathbb Z_2^d}
\chi_a(x)|\widehat a\rangle.
\]

\end{enumerate}
\end{proposition}

\begin{proof}
Properties (i) and (ii) follow directly from arithmetic over
$\mathbb Z_2$. In particular,
\[
\chi_a(x\oplus y)
=
(-1)^{a\cdot(x\oplus y)}
=
(-1)^{a\cdot x}(-1)^{a\cdot y}
=
\chi_a(x)\chi_a(y),
\]
and similarly
\[
\chi_a(x)\chi_b(x)
=
(-1)^{(a\oplus b)\cdot x}
=
\chi_{a\oplus b}(x).
\]

For (iii), property (ii) gives
\[
\sum_{x\in\mathbb Z_2^d}
\chi_a(x)\chi_b(x)
=
\sum_{x\in\mathbb Z_2^d}
\chi_{a\oplus b}(x).
\]
If $a=b$, then $a\oplus b=0$, so every term equals $1$ and the
sum is $2^d$. If $a\neq b$, then $a\oplus b\neq0$. A nontrivial
Walsh character takes the values $+1$ and $-1$ equally often, and
hence the sum vanishes. This also proves (v).

Property (iv) follows by the same argument with the roles of the
group and character variables interchanged:
\[
\sum_{a\in\mathbb Z_2^d}
\chi_a(x)\chi_a(y)
=
\sum_{a\in\mathbb Z_2^d}
\chi_a(x\oplus y)
=
2^d\delta_{x,y}.
\]

It remains to verify (vi). By (iii),
\begin{align*}
\langle\widehat a|\widehat b\rangle
&=
\frac{1}{2^d}
\sum_{x\in\mathbb Z_2^d}
\chi_a(x)\chi_b(x)\\
&=
\delta_{a,b},
\end{align*}
so the states $|\widehat a\rangle$ form an orthonormal basis.

To obtain the inverse transform, substitute the definition of
$|\widehat a\rangle$:
\begin{align*}
\frac{1}{2^{d/2}}
\sum_{a\in\mathbb Z_2^d}
\chi_a(x)|\widehat a\rangle
&=
\frac{1}{2^d}
\sum_{a\in\mathbb Z_2^d}
\sum_{y\in\mathbb Z_2^d}
\chi_a(x)\chi_a(y)|y\rangle\\
&=
\sum_{y\in\mathbb Z_2^d}
\left(
\frac{1}{2^d}
\sum_{a\in\mathbb Z_2^d}
\chi_a(x)\chi_a(y)
\right)|y\rangle\\
&=
\sum_{y\in\mathbb Z_2^d}
\delta_{x,y}|y\rangle\\
&=
|x\rangle,
\end{align*}
where the third equality uses property (iv).
\end{proof}

\begin{lemma}[Fourier block decomposition]\label{lem:block}
In the Fourier basis of $\mathcal H_P$, the walk operator is unitarily
equivalent to
\[
 U\cong\bigoplus_{a\in\Z_2^d}U_a,
 \qquad U_a=D_aC,
\]
where
\begin{equation}\label{eq:Da}
 D_a\ket\omega=\chi_a(\omega)\ket\omega
 \qquad(\omega\in\Omega).
\end{equation}
\end{lemma}

\begin{proof}
For each $a\in\Z_2^d$ and $\omega\in\Omega$,
\begin{align*}
 S(\ket\omega\otimes\ket{\widehat a})
 &=S(\ket\omega\otimes 2^{-d/2}\sum_{x \in \Z_2^d}\chi_a(x)\ket{x})\\
 &=2^{-d/2}\sum_{x \in \Z_2^d}\chi_a(x)S(\ket\omega\otimes \ket{x})\\
 &=2^{-d/2}\sum_x\chi_a(x)\ket{\omega,x\oplus\omega}\\
 &=2^{-d/2}\sum_y\chi_a(y\oplus\omega)\ket{\omega,y}\\
 &=2^{-d/2}\sum_y\chi_a(y)\chi_a(\omega)\ket{y} \qquad \text{(From property i)  }\\
 &=\chi_a(\omega)\ket\omega \otimes 2^{-d/2}\sum_y\chi_a(y)\ket{y}\\
 &=\chi_a(\omega)\ket\omega\otimes\ket{\widehat a} \qquad \text{(From property vi)  }.
\end{align*}
Thus the shift restricted to Fourier mode $a$ is $D_a$.  The coin does not
act on position, so the restriction of $U=S(C\otimes\Id)$ to that mode is
$D_aC$.
\end{proof}

For each $a$, define
\begin{equation}\label{eq:Mma}
 M(a):=\sum_{\omega\in\Omega}\chi_a(\omega),
 \qquad
 m(a):=\frac{M(a)}\Delta.
\end{equation}
Let $\ket{F_a}:=D_a\ket D$. Then,
\begin{equation}
\ket{F_a} = D_a\ket D = \frac{1}{\sqrt{\Delta}}\sum_{\omega \in \Omega}D_a  \ket\omega=\frac{1}{\sqrt{\Delta}}\sum_{\omega \in \Omega}\chi_a(\omega)  \ket\omega  \qquad \text{(From lemma \ref{lem:block})}
\end{equation}
And we get
\begin{equation}
    \braket{D}{F_a} = \frac{1}{\Delta}\sum_{\omega,\omega' \in \Omega }\chi_a(\omega) \braket{\omega}{\omega'}=\frac{1}{\Delta}\sum_{\omega\in \Omega }\chi_a(\omega) =m(a)
\end{equation}
The vector $\ket{F_a}$ can be decomposed as
\begin{equation}
|F_a\rangle
=
\underbrace{m(a)|D\rangle}_{\text{parallel to }|D\rangle}
+
\underbrace{\left(|F_a\rangle-m(a)|D\rangle\right)}_{\text{orthogonal to }|D\rangle}.
\end{equation}
\[
|F_a\rangle
=
m(a)|D\rangle
+
s(a)|E_a\rangle,
\]
where $s(a) = \sqrt{1-m(a)^2}$ and
\begin{equation}
|E_a\rangle
:=
\frac{|F_a\rangle-m(a)|D\rangle}
{s(a)}.
\end{equation}
$\ket{E_a}$ is the normalized component of $\ket{F_a}$ that is orthogonal to $\ket{D}$

\begin{lemma}[Two-dimensional reduction]\label{lem:rotation}
 The subspace
\[
 \mathcal K_a:=\operatorname{span}\{\ket D,\ket{F_a}\}
\]
is invariant under $U_a$.  If $|m(a)|<1$, in the ordered orthonormal basis $\{\ket D,\ket{E_a}\}$,
\begin{equation}\label{eq:rotationmatrix}
 U_a\big|_{\mathcal K_a}=
 \begin{pmatrix}
 m(a)&-s(a)\\
 s(a)&m(a)
 \end{pmatrix}.
\end{equation}
Thus, restricted to the subspace $\mathcal K_a$ the action of $U_a$ is a rotation  with
\begin{equation}\label{eq:theta}
 \theta_a:=\arccos m(a)\in[0,\pi],
\end{equation}
Consequently we have for every integer $T\ge0$,
\begin{equation}\label{eq:cosformula}
 \bra D U_a^T\ket D=\cos(T\theta_a).
\end{equation}

\end{lemma}

\begin{proof}
Since $C\ket D=\ket D$,
\begin{equation}
\label{eq:Rot1}
 U_a\ket D=D_aC\ket D=D_a\ket D=\ket{F_a}
 =m(a)\ket D+s(a)\ket{E_a}
 \end{equation}
Because $\ket{E_a}\perp\ket D$, the Grover coin which is a reflection acts on it as $C\ket{E_a} = -1 \ket{E_a}$.
Using $D_a^2=\Id$ and
$\ket{E_a}=(D_a\ket D-m(a)\ket D)/s(a)$, we obtain
\begin{align*}
 U_a\ket{E_a}
 &=D_aC\ket{E_a}\\
 &= -D_a\ket{E_a}\\
 &= -D_a\Big(\frac{|F_a\rangle-m(a)|D\rangle}
{s(a)}\Big)\\
 &= -D_a\Big(\frac{D_a|D\rangle-m(a)|D\rangle}
{s(a)}\Big)\\
 &=-\Big(\frac{\ket D-m(a)\ket{F_a}}{s(a)}\Big)\\
&=-\Big(\frac{\ket D-m(a)(m(a)\ket{D}+s(a)\ket{E_a})}{s(a)}\Big)\\
&=-\Big(\frac{(1-m(a)^2)\ket D-m(a)s(a)\ket{E_a})}{s(a)}\Big)
\end{align*}
So
\begin{equation}
\label{eq:Rot2}
U_a\ket{E_a} =-s(a)\ket D+m(a)\ket{E_a}.
\end{equation}

From equations \eqref{eq:Rot1} and  \eqref{eq:Rot2} we get \eqref{eq:rotationmatrix}.  It is the rotation through angle
$\theta_a$ with cosine $m(a)$, and \eqref{eq:cosformula} follows.
\end{proof}
Note that the same formula remains valid when $m(a)=\pm1$ in which case the subspace $\mathcal K_a:=\operatorname{span}\{\ket D,\ket{F_a}\}$ is one dimensional.

\section{Target amplitude and universal identities}
Define the amplitude at position $v$ in the uniform coin direction by
\begin{equation}\label{eq:ATdef}
 A_T(v):=(\bra D\otimes\bra v)U^T\ket{\psi_0}.
\end{equation}

\begin{proposition}[Fourier formula for the target amplitude]\label{prop:amplitude}
For every $v\in\Z_2^d$ and every integer $T\ge0$,
\begin{equation}\label{eq:AT}
 A_T(v)=2^{-d}\sum_{a\in\Z_2^d}
 \chi_a(v)\cos(T\theta_a).
\end{equation}
Furthermore,
\begin{equation}\label{eq:lowerbound}
 p_T(v)\ge |A_T(v)|^2.
\end{equation}
\end{proposition}

\begin{proof}
Since
\begin{equation*}
    \ket{\hat{a}} = 2^{-d/2} \sum_{x\in \Z_2^d} \chi_a(x) \ket{x}
\end{equation*}

\begin{equation}
\label{eq:ip}
    \braket{v}{\hat{a}} = 2^{-d/2} \chi_a(v) 
\end{equation}

Moreover, the Fourier inversion formula gives
\begin{equation*}
\ket{x} = 2^{-d/2} \sum_{a\in Z_2^d}\chi_a(x)\ket{\hat{a}}
\end{equation*}
We get
\begin{eqnarray}
\label{eq:inversion}
\ket{0} &= 2^{-d/2} \sum_{a\in Z_2^d}\chi_a(0)\ket{\hat{a}}\\ \nonumber
\ket{0} &= 2^{-d/2} \sum_{a\in Z_2^d}\ket{\hat{a}}
\end{eqnarray}
We have 
\begin{align*}
    A_T(v) &= \bra{D,v}U^T\ket{D,0}        
\end{align*}
Substituting  \eqref{eq:inversion} we get
\begin{align*}
 A_T(v) &=   2^{-d/2} \sum_{a\in Z_2^d}\bra{D,v}U^T\ket{D,\hat{a}}\\ 
 &=   2^{-d/2} \sum_{a\in Z_2^d}\bra{D,v}\bigoplus_a U_a^T\ket{D,\hat{a}}\\ 
 &=   2^{-d/2} \sum_{a\in Z_2^d}\bra{D}U_a^T\ket{D}\braket{v}{\hat{a}}
\end{align*}
The last step is because $U_a$ acts non trivially only on $\ket{D,\hat{a}}$. Now substituting equation \eqref{eq:ip} we get
\begin{equation}
  A_T(v) =  2^{-d} \sum_{a\in Z_2^d} \chi_a(v)\cos(T\theta_a) 
\end{equation}
$p_T(v)$ is the probability of observing position $v$ without measuring the coin.
Let
\begin{equation}
    \ket{\phi_T(v)} := (\Id_{\mathcal H_C} \otimes \bra{v})U^T\ket{\psi_0}
\end{equation}
Then,
\begin{equation}
     p_T(v) = \|  \ket{\phi_T(v)}\|^2 
\end{equation}
On the other hand,
\begin{equation}
    A_T(v) = \braket{D}{\phi_T(v)}
\end{equation}
Since $\|D\|^2=1$ Cauchy-Schwarz gives
\begin{equation}
    |A_T(v)|^2 = |\braket{D}{\phi_T(v)}|^2 \leq \|D\|^2 \|\phi_T(v)\|^2 =  \|\phi_T(v)\|^2
\end{equation}
Hence
\begin{equation}
    p_T(v) \geq |A_T(v)|^2 
\end{equation}
\end{proof}

\begin{proposition}[Product-phase identity]\label{prop:phase}
For the canonical target $\sigma$ in \eqref{eq:sigma},
\begin{equation}\label{eq:phaseidentity}
 \chi_a(\sigma)=(-1)^{(\Delta-M(a))/2}
 \qquad(a\in\Z_2^d).
\end{equation}
\end{proposition}

\begin{proof}
Let $r(a)$ be the number
of negative signs among the $\Delta$ values. 
\begin{equation}
    r(a) = |\{\omega \in \Omega: \chi_a(\omega) = -1\}|
\end{equation}
Then
$M(a)=\Delta-2r(a)$, while multiplicativity of characters gives
\[
 \chi_a(\sigma)=\prod_{\omega\in\Omega}\chi_a(\omega)=(-1)^{r(a)}.
\]
Since $r(a)=(\Delta-M(a))/2$, the claim follows.
\end{proof}

\begin{lemma}[Universal second moment]\label{lem:secondmoment}
If the elements of $\Omega$ are distinct, then
\begin{equation}\label{eq:secondmoment}
 2^{-d}\sum_{a\in\Z_2^d}M(a)^2=\Delta.
\end{equation}
Consequently, for every $L>0$,
\begin{equation}\label{eq:tail}
 2^{-d}\left|\{a:|M(a)|>L\}\right|\le\frac\Delta{L^2}.
\end{equation}
\end{lemma}

\begin{proof}
Expanding and using character orthogonality,
\begin{align*}
 \sum_a M(a)^2
 &=\sum_a \sum_{\omega,\omega'\in\Omega}
   \chi_a(\omega)\chi_a(\omega')\\
 &=\sum_{\omega,\omega'\in\Omega}\sum_a
   \chi_a(\omega\oplus\omega')\\
 &=2^d\sum_{\omega,\omega'\in\Omega}
   \mathbf 1_{\omega=\omega'}
 =2^d\Delta.
\end{align*}
This proves \eqref{eq:secondmoment}; \eqref{eq:tail} is Markov's inequality
applied to $M(a)^2$.
\end{proof}

\section{One-shot hitting time theorem}
We now prove the hitting time theorem for cubelike graphs.  The proof keeps all estimates uniform in
$d$ and in the choice of $\Omega$.

\begin{theorem}[ One-shot hitting time on cubelike graphs]
\label{thm:main}
Let
\[
G=\Cay(\mathbb Z_2^d,\Omega)
\]
be a cubelike graph with degree $\Delta=|\Omega|$, and consider the
Grover-coined moving-shift quantum walk starting from
\[
|\psi_0\rangle=|D\rangle\otimes|0\rangle.
\]
Define the canonical target vertex by
\[
\sigma=\bigoplus_{\omega\in\Omega}\omega.
\]

Choose an integer $T$ having the same parity as $\Delta$ and satisfying
\[
\left|T-\frac{\pi\Delta}{2}\right|\leq 1.
\]
Then, for any family of such graphs with $\Delta\to\infty$,
\begin{equation}
\label{eq:mainbound}
p_T(\sigma)
=
1-O(\Delta^{-1/5}).    
\end{equation}

In particular, after $T=\Theta(\Delta)$ steps, a measurement of the
position finds the walker at $\sigma$ with probability tending to one
as $\Delta\to\infty$.
\end{theorem}

\begin{proof}

Let
\begin{equation}\label{eq:eta}
 \kappa_{\Delta,T}:=(-1)^{(\Delta-T)/2}\in\{\pm1\}.
\end{equation}
We prove that $A_T(\sigma)=\kappa_{\Delta,T}+O(\Delta^{-1/5})$.

\medskip
\noindent\textbf{Step 1: Typical set of Fourier modes.}
Choose
\begin{equation}\label{eq:L}
 L:=\Delta^{3/5},
 \qquad
 \mathcal G:=\{a\in\Z_2^d:|M(a)|\le L\}.
\end{equation}
By Lemma~\ref{lem:secondmoment},
\begin{equation}\label{eq:badfraction}
 2^{-d}|\mathcal G^c|\le\frac\Delta{L^2}=\Delta^{-1/5}.
\end{equation}
For $a\in\mathcal G$,
$|M(a)|/\Delta\le\Delta^{-2/5}$, which is at most $1/2$ for all sufficiently
large $\Delta$.

\medskip
\noindent\textbf{Step 2: Rewrite $\cos(T\theta_a)$}

Since
\begin{align*}
    \cos(\theta_a) &= \frac{M(a)}{\Delta}\\
    \sin(\frac{\pi}{2} -\theta_a)&= \frac{M(a)}{\Delta}\\
    \theta_a &=\frac{\pi}{2} -\sin^{-1}(\frac{M(a)}{\Delta})\\
    T\theta_a &=\frac{T\pi}{2} -T\sin^{-1}(\frac{M(a)}{\Delta})
\end{align*}
Add and subtract $\frac{\pi M(a)}{2}$
\begin{equation*}
    T \theta_a = \pi\Big[\frac{T-M(a)}{2}\Big] +\Big[\frac{\pi M(a)}{2}-T\sin^{-1}(\frac{M(a)}{\Delta})\Big]
\end{equation*}
Define 
\[\rho_a := \frac{\pi M(a)}{2}-T\sin^{-1}(\frac{M(a)}{\Delta})\]
Now $T-M(a)$ is an  even integer because $T$ has the same parity as $\Delta$ and by proposition \ref{prop:phase} $\Delta -M(a)$ is even we get
\begin{equation}
    \cos(T\theta_a) = (-1)^{\frac{T-M(a)}{2}}\cos \rho_a
\end{equation}

\medskip
\noindent\textbf{Step 3: Combine with character phase.}
\medskip

By proposition \ref{prop:phase}
\begin{equation*}
    \chi_a(\sigma) = (-1)^\frac{\Delta-M(a)}{2}
\end{equation*}
\begin{align*}
    \chi_a(\sigma) \cos(T\theta_a) &= (-1)^{\frac{\Delta-M(a)}{2}}(-1)^{\frac{T-M(a)}{2}}\cos \rho_a \\
    \chi_a(\sigma) \cos(T\theta_a) &= (-1)^{\frac{\Delta+T}{2}-M(a)}\cos \rho_a
\end{align*}
Now,
\begin{equation*}
    \frac{\Delta+T}{2}-M(a) -\frac{T-\Delta}{2}= \Delta-M(a)
\end{equation*}
Since $\Delta-M(a)=2r(a)$  is  even, the two powers of $-1$ agree
\begin{equation*}
    (-1)^{(\frac{\Delta+T}{2}-M(a))} = (-1)^{\frac{\Delta-T}{2}}
\end{equation*}
So
\begin{align}
\label{eq:aligned}
     \chi_a(\sigma) \cos(T\theta_a) &= (-1)^{\frac{\Delta-T}{2}}\cos\rho_a\\ \nonumber
     \chi_a(\sigma) \cos(T\theta_a) &= \kappa_{\Delta,T}\cos\rho_a
\end{align}
Note that $\kappa_{\Delta,T}$ is independent of $a$

\medskip
\noindent\textbf{Step 4: Uniform Taylor estimate for $\rho_a$ on good modes.}
For $a \in \mathcal G,\, \frac{|M(a)|}{\Delta}\leq \Delta^{-2/5}$. Since
\begin{align*}
    \sin^{-1}(x) &= x +O(x^3) \qquad \text{for $|x| \leq \frac{1}{2}$}\\
    \sin^{-1}(\frac{M(a)}{\Delta}) &= \frac{M(a)}{\Delta} +O\Big(\frac{M(a)^3}{\Delta^3} \Big)
\end{align*}

Also 
\begin{align*}
    T= \frac{\pi\Delta}{2} +\eta \quad |\eta| \leq 1
\end{align*}

\begin{align*}
    \rho_a &= \frac{\pi M(a)}{2} - T\sin^{-1}\Big(\frac{M(a)}{\Delta}\Big)\\
    &= \frac{\pi M(a)}{2} - \Big(\frac{\pi\Delta}{2} +\eta\Big)\Big[\frac{M(a)}{\Delta} +O\Big(\frac{M^3}{\Delta^3}\Big) \Big]\\
    &= \frac{-\eta M(a)}{\Delta} +O\Big(\frac{|M(a)|^3}{\Delta^2}\Big)
\end{align*}
On a good mode $|M(a)| \leq \Delta^{3/5}$ and $\frac{|M(a)|}{\Delta} \leq \Delta^{-2/5}$  and 
\begin{equation*}
\frac{|M(a)|^3}{\Delta^2} = \frac{\Delta^{9/5}}{\Delta^2}= \Delta^{-1/5}
\end{equation*}
So since the Taylor remainder  is uniform for $x\leq \frac{1}{2}$ is we obtain
\begin{align*}
    |\rho_a| = O(\Delta^{-1/5})
\end{align*}
uniformly for $a\in \mathcal{G}$

\medskip
\noindent\textbf{Step 5: Average of the Fourier modes.}
By Proposition~\ref{prop:amplitude},
\[
 A_T(\sigma)=2^{-d}\sum_{a\in\mathcal G}
 \chi_a(\sigma)\cos(T\theta_a)
 +2^{-d}\sum_{a\notin\mathcal G}
 \chi_a(\sigma)\cos(T\theta_a).
\]
The second sum has absolute value at most
$2^{-d}|\mathcal G^c|=O(\Delta^{-1/5})$ by \eqref{eq:badfraction}.
Substituting \eqref{eq:aligned} in the first sum, yields
\begin{equation}\label{eq:Aestimate}
 A_T(\sigma)=2^{-d}\sum_{a\in\mathcal G}\kappa_{\Delta,T}\cos(\rho_a)+O(\Delta^{-1/5}).
\end{equation}
Furthermore on good modes we have seen that $|\rho_a| = O(\Delta^{-1/5})$ and for sufficiently large $\Delta$  we have $\cos(\rho_a) = 1 -O(\rho_a^2) = 1- O(\Delta^{-2/5})$ and we get
\begin{align*}\label{eq:Aestimate}
 A_T(\sigma)&=\kappa_{\Delta,T}2^{-d}\sum_{a\in\mathcal G}(1-O(\Delta^{-2/5)})+O(\Delta^{-1/5})\\
A_T(\sigma)&=\kappa_{\Delta,T}(1-O(\Delta^{-2/5)})2^{-d}\sum_{a\in\mathcal G}1+O(\Delta^{-1/5})\\
A_T(\sigma)&=\kappa_{\Delta,T}(1-O(\Delta^{-2/5})(1-\frac{|\mathcal G^c|}{2^d})+O(\Delta^{-1/5})\\
A_T(\sigma)&=\kappa_{\Delta,T}(1-O(\Delta^{-2/5})(1-O(\Delta^{-1/5})+O(\Delta^{-1/5})\\
\end{align*}
Which gives
\begin{equation}
A_T(\sigma)=\kappa_{\Delta,T}+O(\Delta^{-1/5})\\
\end{equation}
Hence
\[
 |A_T(\sigma)|^2=1-O(\Delta^{-1/5}).
\]
Finally, \eqref{eq:lowerbound} gives
$p_T(\sigma)\ge|A_T(\sigma)|^2$.  Since probabilities are at most one,
\eqref{eq:mainbound} follows.
\end{proof}

\section{Concurrent hitting time theorem}
We first establish a general relation between one-shot and concurrent
detection probabilities. The following lemma is independent of the
underlying graph and of the particular structure of the quantum walk;
it requires only unitary evolution and repeated projective measurement
of the target subspace.

\begin{lemma}[One-shot and concurrent detection probabilities]
\label{lem:oneshot-concurrent}
Let $U$ be a unitary operator on a finite-dimensional Hilbert space,
let $P$ be an orthogonal projector onto a target subspace, and let
$Q=I-P$. For an initial state $|\psi_0\rangle$, define
\[
p_T=\|PU^T|\psi_0\rangle\|^2
\]
and
\[
H_T^{\mathrm{Conc}}
=
\sum_{t=1}^T
\|PU(QU)^{t-1}|\psi_0\rangle\|^2.
\]
Then, for every integer $T\geq1$,
\[
p_T\leq T H_T^{\mathrm{Conc}}.
\]
\end{lemma}
\begin{proof}
For $t=1,\ldots,T$, define the unnormalized first-detection state
\[
|\phi_t\rangle
=
PU(QU)^{t-1}|\psi_0\rangle .
\]
Its squared norm is the probability that the target is detected for
the first time at step $t$:
\[
q_t=\|\phi_t\|^2.
\]
Hence
\[
H_T^{\mathrm{Conc}}
=
\sum_{t=1}^T q_t.
\]

We first relate the unmeasured state at time $T$ to the first-detection
states. Repeatedly inserting
\[
I=P+Q
\]
between successive applications of $U$ gives
\begin{align*}
U^T|\psi_0\rangle &=U^{T-1}(P+Q)U|\psi_0\rangle \\
&=U^{T-1}PU\ket{\psi_0} + U^{T-1}QU\ket{\psi_0}\\
&= U^{T-1}PU\ket{\psi_0} + U^{T-2}(P+Q)UQU\ket{\psi_0}\\
&= U^{T-1}PU\ket{\psi_0} + U^{T-2}PUQU\ket{\psi_0} +U^{T-2}(QU)^2\ket{\psi_0}
\end{align*}
Continuing in this manner we obtain
\begin{equation}
U^T|\psi_0\rangle
=
\sum_{t=1}^T
U^{T-t}PU(QU)^{t-1}|\psi_0\rangle
+
(QU)^T|\psi_0\rangle.
\end{equation}
Multiplying by $P$ from the left, the last term vanishes because
$PQ=0$. Therefore,
\begin{equation}
PU^T|\psi_0\rangle = \sum_{t=1}^T PU^{T-t}PU(QU)^{t-1}|\psi_0\rangle =
\sum_{t=1}^T
PU^{T-t}|\phi_t\rangle.
\end{equation}

Taking norms and using the triangle inequality,
\[
\|PU^T|\psi_0\rangle\|
\leq
\sum_{t=1}^T
\|PU^{T-t}|\phi_t\rangle\|.
\]
Because $P$ is an orthogonal projection and $U$ is unitary,
\[
\|PU^{T-t}|\phi_t\rangle\|
\leq
\|U^{T-t}|\phi_t\rangle\|
=
\|\phi_t\|.
\]
Hence
\[
\sqrt{p_T}
=
\|PU^T|\psi_0\rangle\|
\leq
\sum_{t=1}^T \|\phi_t\|
=
\sum_{t=1}^T \sqrt{q_t}.
\]

Applying the Cauchy--Schwarz inequality,
\[
\left(
\sum_{t=1}^T \sqrt{q_t}
\right)^2
\leq
T\sum_{t=1}^T q_t.
\]
Squaring the previous inequality therefore yields
\[
p_T
\leq
T\sum_{t=1}^T q_t
=
T H_T^{\mathrm{Conc}}.
\]
\end{proof}

We now prove the concurrent hitting time theorem for cubelike graphs.  
\begin{theorem}[Concurrent hitting time theorem]
\label{thm:Concurrent}
Let $G=\Cay(\mathbb{Z}_2^d,\Omega)$ be a cubelike graph of degree $\Delta$. Let
\begin{equation*}
\sigma =\bigoplus_{\omega \in \Omega} \omega
\end{equation*}
Consider the Grover-coined moving-shift walk on $G$, with the target
measured after every step using the projectors
\[
P=I\otimes|\sigma\rangle\langle\sigma|,
\qquad Q=I-P.
\]
Start the walk from $\ket{\psi_0} = \ket{D}\otimes \ket{0}$. Choose $T$ such that $T\equiv \Delta \mod 2$ and $\Big|T-\frac{\pi\Delta}{2}\Big| \leq 1$, then the probability of detecting $\sigma$ by time $T$ satisfies
\begin{equation}
    H_T^{\text{Conc}}(\sigma) = \Omega\Big(\frac{1}{\Delta}\Big)
\end{equation}
\end{theorem}
\begin{proof}
By Theorem~\ref{thm:main},
\[
p_T(\sigma)=1-O(\Delta^{-1/5}).
\]
Applying Lemma~\ref{lem:oneshot-concurrent} 
\[
H_T^{\mathrm{Conc}}(\sigma)
\ge
\frac{p_T(\sigma)}{T}.
\]
Since
\[
T=\frac{\pi\Delta}{2}+O(1),
\]
\begin{equation*}
    \frac{1}{T} = \frac{2}{\pi\Delta} + O(\Delta^{-2})
\end{equation*}
So,
\begin{equation*}
H_T^{\mathrm{Conc}}(\sigma)
\geq \Big(1-O(\Delta^{-1/5})\Big) \Big(\frac{2}{\pi\Delta} + O(\Delta^{-2})\Big)
\end{equation*}
Therefore,
\begin{equation}
    H_T^{\mathrm{Conc}}(\sigma) = \Omega(\frac{1}{\Delta})
\end{equation}
\end{proof}
The concurrent-hitting bound also gives an operational interpretation.
A single measured walk of length $T=\Theta(\Delta)$ detects the target
with probability at least $\Omega(\Delta^{-1})$. Consequently, by
restarting the walk independently $O(\Delta)$ times, the probability of
detecting the target can be amplified to a constant bounded away from
zero, using a total of $O(\Delta^2)$ walk steps. More generally,
$O(\Delta\log(1/\varepsilon))$ independent runs suffice to obtain
success probability at least $1-\varepsilon$, corresponding to
$O(\Delta^2\log(1/\varepsilon))$ total walk steps.

For the $d$-dimensional hypercube, where $\Delta=d$, the preceding
bound gives
\[
H_T^{\mathrm{Conc}}(\sigma)=\Omega(1/d).
\]
Kempe's original analysis established the weaker lower bound
$\Omega(1/(d\log^2 d))$ within a linear number of steps
\cite{Kempe2005}. Thus, for the hypercube, the present argument removes
the logarithmic loss in Kempe's concurrent-hitting bound. We are not
aware of a previously published analytic $\Omega(1/d)$ bound for
Kempe's concurrent-hitting formulation.

\section{Remarks and Examples}

We collect several consequences and examples that illustrate the scope of Theorems \ref{thm:main} and \ref{thm:Concurrent}. A notable feature of the argument is that the detailed
geometry of the cubelike graph does not enter the proof. 

\subsection{Parity of the degree}

The  theorem applies uniformly to both even and odd degrees.
Indeed,
\[
M(a)=\Delta-2r(a),
\]
where
\[
r(a)=|\{\omega\in\Omega:\chi_a(\omega)=-1\}|.
\]
Consequently,
\[
M(a)\equiv\Delta\pmod 2.
\]
Choosing the observation time \(T\) so that
\[
T\equiv\Delta\pmod 2
\]
therefore guarantees
\[
\frac{T-M(a)}{2}\in\mathbb Z
\]
for every Fourier mode \(a\). This is the parity condition needed in the
phase-alignment argument.

Thus the relevant distinction is not between even- and odd-degree
cubelike graphs. Rather, the observation time must be chosen with the
same parity as the degree. In both cases one may choose \(T\) satisfying
\[
\left|T-\frac{\pi\Delta}{2}\right|\le 1.
\]

\subsection{The canonical target need not be an antipode}

The distinguished vertex
\[
\sigma=\bigoplus_{\omega\in\Omega}\omega
\]
is defined algebraically and need not coincide with a graph-theoretic
antipode or a vertex at maximum distance from the origin. The role of
\(\sigma\) in the proof follows instead from the character identity
\[
\chi_a(\sigma)
=
(-1)^{(\Delta-M(a))/2},
\]
which precisely cancels the mode-dependent sign produced by the quantum
walk at the observation time \(T\).

Thus the theorem should be interpreted as a statement about coherent
transport to an algebraically distinguished vertex. In special families,
such as the hypercube, this vertex is also the usual antipode.

\subsection{The case \texorpdfstring{\(\sigma=0\)}{sigma=0}}

It is possible that
\[
\bigoplus_{\omega\in\Omega}\omega=0.
\]
In this case the theorem gives asymptotically perfect return to the
initial position rather than transport to a distinct vertex:
\[
p_T(0)=1-O(\Delta^{-1/5}).
\]

The  theorem remains valid without modification. For concurrent
hitting, however, one should interpret the result as a first-return
problem, with measurements beginning after the first step rather than at
time zero.

\subsection{Hypercubes}

The \(d\)-dimensional hypercube is obtained from
\[
\Omega=\{e_1,\ldots,e_d\},
\]
where \(e_i\) denotes the \(i\)-th standard basis vector of
\(\mathbb Z_2^d\). Here
\[
\Delta=d
\]
and
\[
\sigma=e_1\oplus\cdots\oplus e_d=(1,\ldots,1).
\]
Thus the canonical target is exactly the usual antipodal vertex.

For a character indexed by \(a\in\mathbb Z_2^d\),
\[
M(a)
=
\sum_{j=1}^d(-1)^{a_j}
=
d-2|a|,
\]
where \(|a|\) is the Hamming weight of \(a\). Hence
\[
\cos\theta_a
=
1-\frac{2|a|}{d}.
\]

Theorem~7 therefore recovers the hypercube phase-alignment mechanism
without using the explicit binomial multiplicities of the Hamming-weight
classes. If \(T\) is the integer satisfying
\[
T\equiv d\pmod2,
\qquad
\left|T-\frac{\pi d}{2}\right|\le1,
\]
then
\[
p_T(1^d)=1-O(d^{-1/5}).
\]

The estimate is weaker than results that exploit the full permutation
symmetry of the hypercube, but the present proof has the advantage that
it extends unchanged to arbitrary cubelike generating sets.

\subsection{Augmented cubes}
Augmented cubes were introduced as an alternative interconnection-network topology to the hypercube, with high connectivity and small diameter
\cite{ChoudumSunitha2002}. These properties are desirable in parallel and distributed networks, where high connectivity provides robustness
against processor or communication-link failures.
The augmented cube \(AQ_n\) considered in~\cite{MulherkarRajdeepakSunitha2022}
is the cubelike graph
\[
AQ_n=\operatorname{Cay}(\mathbb Z_2^n,\Omega_n),
\]
with generating set
\[
\Omega_n
=
\{e_1,\ldots,e_n\}
\cup
\{0^{\,n-i}1^i:2\le i\le n\}.
\]
Its degree is
\[
\Delta=2n-1.
\]

The canonical target is
\[
\sigma_n
=
\bigoplus_{\omega\in\Omega_n}\omega.
\]
Since \(2n-1\) is odd, the parity-matched observation time is the nearest
odd integer \(T_n\) to
\[
\frac{\pi(2n-1)}{2}.
\]
Theorem~7 gives
\[
p_{T_n}(\sigma_n)
=
1-O(n^{-1/5}).
\]

Thus the augmented cube provides a nontrivial example in which the
generating set contains vectors of several different Hamming weights,
yet the same Fourier argument applies without any modification. In
particular, neither permutation symmetry nor bounded generator weight is
required.

The concurrent result similarly gives
\[
H^{\mathrm{Conc}}_{T_n}(\sigma_n)
=
\Omega\!\left(\frac1n\right)
\]
within
\[
T_n=O(n)
\]
steps.
\subsection{Numerical illustration}

To illustrate the general nature of Theorem~\ref{thm:main}, we numerically
simulate the Grover-coined moving-shift walk on three families of
cubelike graphs: hypercubes, augmented cubes, and randomly generated
connected cubelike graphs. For each dimension \(d\), the observation time
\(T\) is chosen to be the integer satisfying
\[
T\equiv\Delta\pmod 2
\]
that is closest to \(\pi\Delta/2\). The target is the canonical vertex
\[
\sigma=\bigoplus_{\omega\in\Omega}\omega.
\]

For the random cubelike examples, the generating set is chosen uniformly
from distinct nonzero vectors subject to spanning \(\mathbb Z_2^d\); its
degree is taken to be \(2d-1\), matching the degree of the augmented cube.
Figure~\ref{fig:numerical-hitting} shows the resulting target probability
\(p_T(\sigma)\).

\begin{figure}[ht]
    \centering
    \includegraphics[width=0.78\textwidth]
        {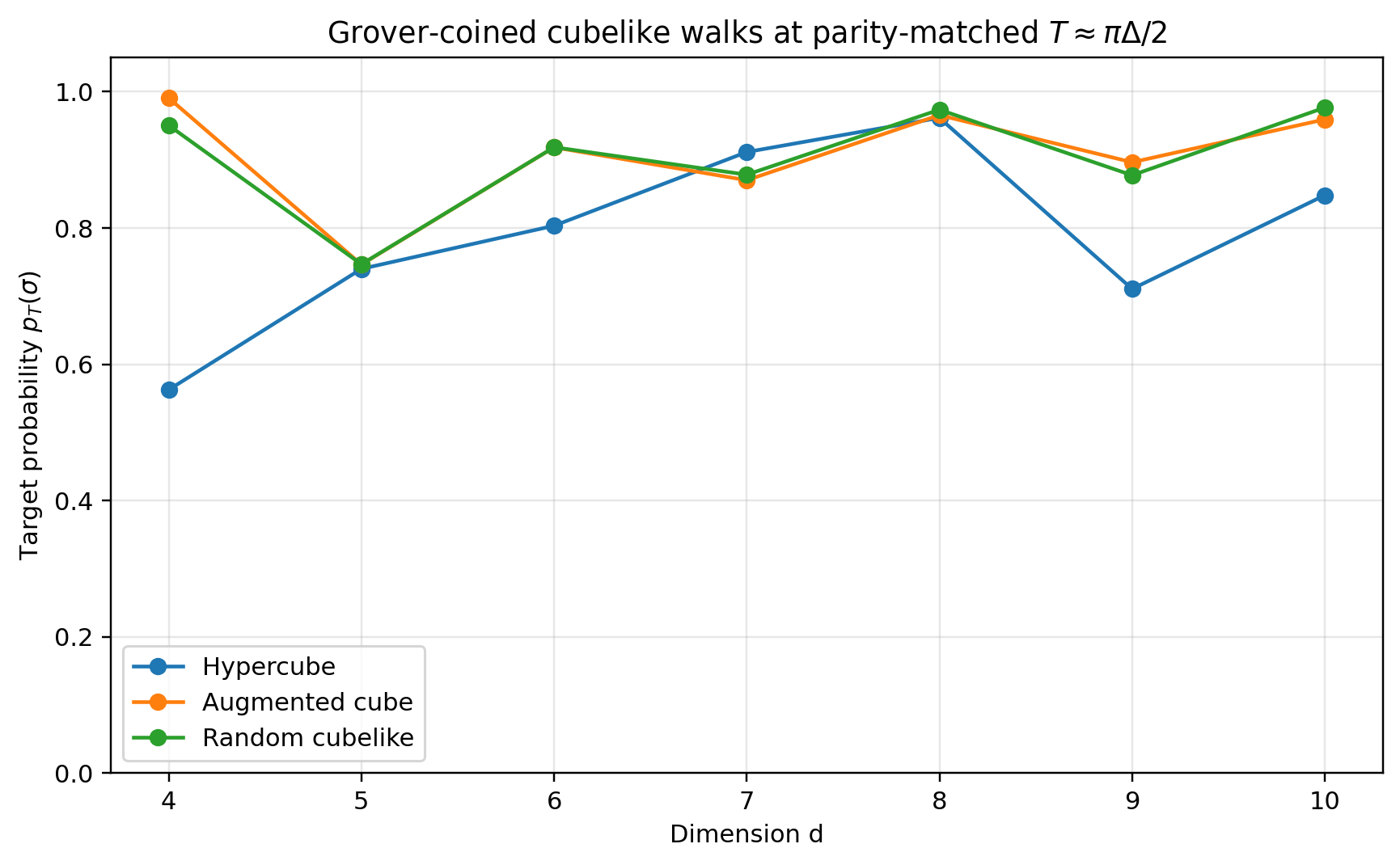}
    \caption{Target probability \(p_T(\sigma)\) for Grover-coined walks
    on hypercubes, augmented cubes, and one randomly generated connected
    cubelike graph in each dimension. The observation time is the
    parity-matched integer nearest to \(\pi\Delta/2\). The numerical
    results illustrate the high target probabilities predicted by
    Theorem~\ref{thm:main}, while also showing finite-size oscillations.}
    \label{fig:numerical-hitting}
\end{figure}

The probabilities need not increase monotonically with dimension, since
the theorem is asymptotic and does not provide a monotonicity statement.
Nevertheless, the examples illustrate that the phase-alignment phenomenon
is not restricted to highly symmetric families such as the hypercube or
the augmented cube.
\subsection{Complete cubelike graphs}

A contrasting example is obtained by taking
\[
\Omega=\mathbb Z_2^d\setminus\{0\}.
\]
The resulting Cayley graph is the complete graph
\[
K_{2^d}.
\]
Writing
\[
N=2^d,
\]
its degree is
\[
\Delta=N-1=2^d-1.
\]

For \(d\ge2\), the XOR of all nonzero elements of
\(\mathbb Z_2^d\) vanishes:
\[
\sigma
=
\bigoplus_{\omega\ne0}\omega
=
0.
\]
Indeed, each coordinate is equal to \(1\) in exactly \(2^{d-1}\) vectors,
which is even for \(d\ge2\).

Thus, for complete cubelike graphs, the canonical hitting theorem
describes a return to the origin rather than transfer to a different
vertex.

The Fourier sums take an especially simple form. For the trivial
character,
\[
M(0)=\Delta,
\]
while for every \(a\ne0\),
\[
M(a)
=
\sum_{\omega\ne0}\chi_a(\omega)
=
-1.
\]
Hence all nontrivial Fourier modes have the same angle
\[
\theta
=
\arccos\!\left(-\frac1\Delta\right).
\]

The uniform-coin return amplitude is therefore
\[
A_T(0)
=
\frac{1}{2^d}
\left[
1+(2^d-1)\cos(T\theta)
\right]
=
\frac{1+\Delta\cos(T\theta)}
{\Delta+1}.
\]
At the parity-matched time
\[
T=\frac{\pi\Delta}{2}+O(1),
\]
the nontrivial Fourier modes again align, giving the asymptotic return
predicted by Theorem~7.

The complete graph also illustrates why the distinction between
transport and return is useful. Although the general theorem detects a
long-time revival near \(T=\pi\Delta/2\), the Grover-coined walk on the
complete graph has additional symmetry and can exhibit much faster
short-time return behavior. Thus the time furnished by Theorem~7 need
not be the earliest time at which a large hitting or return probability
occurs.

\section{Conclusion}

We have studied  one-shot and concurrent hitting for the discrete-time
Grover-coined quantum walk on cubelike graphs
\[
G=\operatorname{Cay}(\mathbb Z_2^d,\Omega),
\]
with degree \(\Delta=|\Omega|\). The main result identifies the algebraically
distinguished vertex
\[
\sigma=\bigoplus_{\omega\in\Omega}\omega
\]
as a natural target for the walk.

For a parity-matched observation time \(T\) satisfying
\[
T\equiv\Delta\pmod 2,
\qquad
\left|T-\frac{\pi\Delta}{2}\right|\le 1,
\]
we proved, for families with \(\Delta\to\infty\),
\[
p_T(\sigma)=1-O(\Delta^{-1/5}).
\]
Thus the probability of finding the walker at the canonical target tends
to one after \(O(\Delta)\) steps. The argument applies uniformly to both
even and odd degree and does not require assumptions on the Hamming
weights of the generators or additional symmetry of the graph.

We also considered the concurrently measured walk, in which the target
is tested after every step. A renewal decomposition together with the
triangle inequality and Cauchy--Schwarz inequality gives
\[
p_T(\sigma)
\le
T H_T^{\mathrm{Conc}}(\sigma).
\]
Combining this with the  result yields
\[
H_T^{\mathrm{Conc}}(\sigma)
\ge
\frac{2}{\pi\Delta}
\left(1-O(\Delta^{-1/5})\right),
\]
and hence
\[
H_T^{\mathrm{Conc}}(\sigma)=\Omega(\Delta^{-1})
\]
within \(T=\Theta(\Delta)\) steps.

The result includes the hypercube as a special case and applies directly
to less symmetric cubelike families such as augmented cubes. When
\(\sigma=0\), as occurs for the complete cubelike graph
\(K_{2^d}\) for \(d\ge2\), the same mechanism produces an asymptotically
perfect return rather than transfer to a distinct vertex. The appearance
of the XOR of the generating set is also closely related to the
continuous-time perfect-state-transfer phenomenon on cubelike graphs,
where the same algebraically distinguished displacement arises.

\subsection*{Future directions}

Several questions remain open. First, the error exponent \(1/5\) arises
from the particular good--bad Fourier-mode decomposition and the use of
only the second moment of \(M(a)\). Sharper information about the
distribution or higher moments of the character sums may improve the
convergence rate and, for particular graph families, may lead to
substantially stronger estimates.

Second, the concurrent result is obtained through the general inequality
\[
H_T^{\mathrm{Conc}}(\sigma)\ge \frac{p_T(\sigma)}{T},
\]
and therefore may not be tight. A direct analysis of the first-detection
amplitudes
\[
PU(QU)^{t-1}|\psi_0\rangle
\]
could determine the actual concurrent hitting probability and expected
hitting time for broader classes of cubelike graphs.

A further direction is to determine whether the phase-alignment
mechanism identified here extends beyond cubelike graphs. More generally,
one may ask which Cayley graphs or other highly structured graph families
possess a distinguished target for which their Fourier or representation-
theoretic phases align at a common time.

Finally, it would be interesting to investigate the corresponding
marked-vertex search problem. The explicit spectral description of the
unperturbed coined walk developed here may provide a starting point for
studying search on general cubelike graphs and for understanding the
relationship between the present hitting phenomenon and quantum-walk
search on the hypercube.
\section*{Acknowledgements and Declaration of Generative AI Use}

During the preparation of this manuscript, the authors used OpenAI's
ChatGPT, including GPT-5-series models, as an interactive research and
writing aid. The tool was used to assist with mathematical brainstorming,
exploration of possible generalizations, checking intermediate algebraic
and logical steps, improving notation and exposition, identifying possible
gaps or inconsistencies in draft proofs, suggesting relevant literature,
and assisting with the preparation and editing of \LaTeX{} text.

All mathematical statements, derivations, and proofs appearing in the
final manuscript were independently reviewed and verified by the authors.
The authors take full responsibility for the correctness and originality
of the results, the accuracy of the citations, and the content of the
manuscript. 

\bibliographystyle{plain}
\bibliography{cubelike_hitting_references}

@article{FarhiGutmann1998,
  author  = {Farhi, E. and Gutmann, S.},
  title   = {Quantum Computation and Decision Trees},
  journal = {Physical Review A},
  volume  = {58},
  pages   = {915--928},
  year    = {1998},
  doi     = {10.1103/PhysRevA.58.915}
}

@inproceedings{Aharonov2001,
  author    = {Aharonov, D. and Ambainis, A. and Kempe, J. and Vazirani, U.},
  title     = {Quantum Walks on Graphs},
  booktitle = {Proceedings of the 33rd Annual ACM Symposium on Theory of Computing},
  pages     = {50--59},
  year      = {2001},
  doi       = {10.1145/380752.380758}
}

@article{Kempe2003,
  author  = {Kempe, J.},
  title   = {Quantum Random Walks: An Introductory Overview},
  journal = {Contemporary Physics},
  volume  = {44},
  number  = {4},
  pages   = {307--327},
  year    = {2003},
  doi     = {10.1080/00107151031000110776}
}

@article{Kempe2005,
  author  = {Kempe, J.},
  title   = {Discrete Quantum Walks Hit Exponentially Faster},
  journal = {Probability Theory and Related Fields},
  volume  = {133},
  number  = {2},
  pages   = {215--235},
  year    = {2005},
  doi     = {10.1007/s00440-004-0423-2}
}

@article{Shenvi2003,
  author  = {Shenvi, N. and Kempe, J. and Whaley, K. B.},
  title   = {Quantum Random-Walk Search Algorithm},
  journal = {Physical Review A},
  volume  = {67},
  pages   = {052307},
  year    = {2003},
  doi     = {10.1103/PhysRevA.67.052307}
}

@article{Ambainis2007,
  author  = {Ambainis, A.},
  title   = {Quantum Walk Algorithm for Element Distinctness},
  journal = {SIAM Journal on Computing},
  volume  = {37},
  number  = {1},
  pages   = {210--239},
  year    = {2007},
  doi     = {10.1137/S0097539705447311}
}

@inproceedings{MooreRussell2002,
  author    = {Moore, C. and Russell, A.},
  title     = {Quantum Walks on the Hypercube},
  booktitle = {Randomization and Approximation Techniques in Computer Science},
  series    = {Lecture Notes in Computer Science},
  volume    = {2483},
  pages     = {164--178},
  publisher = {Springer},
  year      = {2002},
  doi       = {10.1007/3-540-45726-7_14}
}

@article{KroviBrun2006,
  author  = {Krovi, H. and Brun, T. A.},
  title   = {Hitting Time for Quantum Walks on the Hypercube},
  journal = {Physical Review A},
  volume  = {73},
  pages   = {032341},
  year    = {2006},
  doi     = {10.1103/PhysRevA.73.032341}
}

@article{Childs2009,
  author  = {Childs, A. M.},
  title   = {Universal Computation by Quantum Walk},
  journal = {Physical Review Letters},
  volume  = {102},
  pages   = {180501},
  year    = {2009},
  doi     = {10.1103/PhysRevLett.102.180501}
}

@article{BernasconiGodsilSeverini2008,
  author  = {Bernasconi, A. and Godsil, C. and Severini, S.},
  title   = {Quantum Networks on Cubelike Graphs},
  journal = {Physical Review A},
  volume  = {78},
  pages   = {052320},
  year    = {2008},
  doi     = {10.1103/PhysRevA.78.052320}
}

@article{Godsil2012,
  author  = {Godsil, C.},
  title   = {State Transfer on Graphs},
  journal = {Discrete Mathematics},
  volume  = {312},
  number  = {1},
  pages   = {129--147},
  year    = {2012},
  doi     = {10.1016/j.disc.2011.06.032}
}

@article{Friedman2017,
  author  = {Friedman, H. and Kessler, D. A. and Barkai, E.},
  title   = {Quantum Walks: The First Detected Passage Time Problem},
  journal = {Physical Review E},
  volume  = {95},
  pages   = {032141},
  year    = {2017},
  doi     = {10.1103/PhysRevE.95.032141}
}

@article{ChoudumSunitha2002,
  author  = {Choudum, S. A. and Sunitha, V.},
  title   = {Augmented Cubes},
  journal = {Networks},
  volume  = {40},
  number  = {2},
  pages   = {71--84},
  year    = {2002},
  doi     = {10.1002/net.10033}
}

@article{MulherkarRajdeepakSunitha2022,
  author  = {Mulherkar, J. and Rajdeepak, R. and Sunitha, V.},
  title   = {Implementation of Quantum Hitting Times of Cubelike Graphs on {IBM}'s {Qiskit} Platform},
  journal = {International Journal of Quantum Information},
  volume  = {20},
  number  = {07},
  pages   = {2250020},
  year    = {2022},
  doi     = {10.1142/S0219749922500204}
}

@article{VenegasAndraca2012,
  author  = {Venegas-Andraca, S. E.},
  title   = {Quantum Walks: A Comprehensive Review},
  journal = {Quantum Information Processing},
  volume  = {11},
  number  = {5},
  pages   = {1015--1106},
  year    = {2012},
  doi     = {10.1007/s11128-012-0432-5}
}

@book{Portugal2013,
  author    = {Portugal, R.},
  title     = {Quantum Walks and Search Algorithms},
  publisher = {Springer},
  address   = {New York},
  year      = {2013},
  doi       = {10.1007/978-1-4614-6336-8}
}

\end{document}